\documentclass[submission,preliminary,copyright,creativecommons]{eptcs}
\providecommand{\event}{WPTE16} 
\usepackage{breakurl}             
\usepackage{amsmath}
\usepackage{amsthm}
\usepackage{esvect}
\usepackage{paralist}
\usepackage[all]{xy}
\usepackage{tikz}
\usetikzlibrary{positioning,arrows}
\usepackage{pgfplots}

\title{An Environment for Analyzing Space Optimizations in Call-by-Need Functional Languages}
\author{Nils Dallmeyer
\institute{Goethe-University\\Frankfurt am Main}
\email{dallmeyer@ki.cs.uni-frankfurt.de}
\and
 Manfred Schmidt-Schau{\ss}
\institute{Goethe-University\\Frankfurt am Main}
\email{schauss@ki.cs.uni-frankfurt.de}
 }
\def\titlerunning{Analyzing Space Optimizations in Call-by-Need Functional Languages}

\def\authorrunning{N. Dallmeyer \& M. Schmidt-Schau{\ss}}
\begin{document}

\newtheorem{theorem}{Theorem}[section]
\newtheorem{lemma}[theorem]{Lemma} 
\newtheorem{example}[theorem]{Example} 
\newtheorem{definition}[theorem]{Definition} 
\newtheorem{remark}[theorem]{Remark} 
\newtheorem{proposition}[theorem]{Proposition}
\newtheorem{exercise}[theorem]{Exercise}
\newtheorem{corollary}[theorem]{Corollary}
\newtheorem{algo}[theorem]{Algorithm}

\abovedisplayskip3ex plus1ex minus1ex
\newcommand{\env}{\mathit{env}}
\newcommand{\tletrec}{\texttt{letrec}}
\newcommand{\tin}{\texttt{in}}
\newcommand{\tcase}{\texttt{case}}
\newcommand{\tseq}{\texttt{seq}}
\newcommand{\nogc}{\mathit{nogc}}
\newcommand{\wrt}{{w.r.t.\!} }
\newcommand{\eg}{{e.g.\!} }
\newcommand{\ie}{{i.e.\!} }
\newcommand{\letrec}{{\tt letrec}\ }
\newcommand{\lin}{\ {\tt in}\ }
\newcommand{\case}{{\tt case}}
\newcommand{\of}{\ {\tt of}\ }
\newcommand{\True}{{\tt True}}
\newcommand{\False}{{\tt False}}
\newcommand{\Zero}{{\tt Zero}}
\newcommand{\Succ}{{\tt Succ}}
\newcommand{\foldl}{{\tt foldl}}
\newcommand{\foldls}{{\tt foldl'}}
\newcommand{\foldr}{{\tt foldr}}
\newcommand{\reverse}{{\tt reverse}}
\newcommand{\reverselin}{{\tt reverse'}}
\newcommand{\reverselinw}{{\tt reversew}}
\newcommand{\seq}{{\tt seq}\ }
\newcommand{\casepf}{\,\texttt{->}\,}
\newcommand{\ltop}{\text{top}}
\newcommand{\lsub}{\text{sub}}
\newcommand{\lvis}{\text{vis}}
\newcommand{\lnontarg}{\text{nontarg}}
\newcommand{\lvee}{\hspace*{1.25pt}\vee\hspace*{1.25pt}}
\newcommand{\rnoLRP}{\xrightarrow{\LRP}}
\newcommand{\rnortLRP}{\xrightarrow{\LRP,*}}
\newcommand{\rnokLRP}{\xrightarrow{\LRP,k}}
\newcommand{\rnopLRP}{\xrightarrow{\LRP,+}}
\newcommand{\rnogc}{\xrightarrow{nogc}}
\newcommand{\LCSC}{\text{LCSC}}
\newcommand{\RetApp}{\text{\#app}}
\newcommand{\RetSeq}{\text{\#seq}}
\newcommand{\RetCase}{\text{\#case}}
\newcommand{\RetHeap}{\text{\#upd}}
\newcommand{\absm}{Mark 1}
\newcommand{\rln}{\text{rln}}
\newcommand{\rlnall}{\text{rlnall}}
\newcommand{\mln}{\text{mln}}
\newcommand{\mlnall}{\text{mlnall}}
\newcommand{\spmax}{\mathit{spmax}}
\newcommand{\spmaxleq}{\leq_\mathit{maxspace}}
\newcommand{\diverges}{{\uparrow}}
\newcommand{\conv}{{\downarrow}}
\newcommand{\convgc}{{\downarrow}_{nogc}}
\newcommand{\updateChain}{{\sc updateChain}}
\newcommand{\LR}{LR}
\newcommand{\LRP}{LRP}
\newcommand{\LRPgc}{LRPgc}
\newcommand{\LRPi}{\text{LRPi}}
\newcommand{\size}{{\tt size}}
\newcommand{\ignore}[1]{}

\maketitle

\begin{abstract}
We present an implementation of an interpreter LRPi for the call-by-need calculus \LRP{}, based on a variant of Sestoft's abstract machine Mark 1,
extended with an eager garbage collector. 
It is used as a tool for exact space usage analyses as a support for our investigations into  space improvements of call-by-need calculi. 
\end{abstract}

\section{Introduction}
Lazy functional languages like Haskell use call-by-need as evaluation strategy.
This leads to a more declarative way of programming where a specification of the
result is emphasized instead of specifying the sequence of evaluations. This approach allows a lot of correct program transformations
that can potentially be used by a compiler for optimization purposes. 
It would be a very helpful information to know, whether a program transformation decreases time/space  usage 
or in which situations this may occur. 
We will capture this using the notion of improvements. In this paper we  emphasize  space improvements, pursuing     
our long-term research goal to analyze time- and space-improvements for Haskell-like languages. 
The goal of this paper is to put forward further studies on improvements,    
with a main focus on  providing a test environment to support and speed up the analysis of improvements.
 
Previous work on improvements \wrt time usage
(the number of reduction steps), is \eg  \cite{sands91opth,sands95naive,sands95totalcorrectness} 
for call-by-name and  
\cite{imp99,ifl16,ppdp-IB55,schmidt-schauss-sabel-PPDP:2015} for call-by-need.
There seem to be only a few studies on space improvements, by Gustavsson and Sands \cite{gustavssonSands99,gustavssonSands01,gustavssonDiss01}. 
Their notion of (strong) space improvement is mainly the same as ours, however, they use an untyped  (restricted) language. We will investigate a typed language
since typing enables more transformations to be improvements, for example $\texttt{map id xs}$ is equivalent to $\texttt{xs}$ under typing,
 but not in untyped calculi.
  The reason is that also contexts must be typed and thus 
only tests that are (type-)compatible with the intention of the program are used for characterizing improvements. 

We will use the lazy typed functional core language \LRP{} \cite{schmidt-schauss-sabel-WPTE:14} for defining and analyzing space improvements.
\LRP{}  has a rich syntax including \tletrec, data constructors and \tcase-expressions, Haskell's \tseq-operator,
and polymorphic typing, modeling Haskell's core language.  Evaluation in \LRP{} is defined by a rewriting semantics.
An improvement \wrt a measure  is  a locally applicable (and   correct) transformation that transforms an expression $e_1$ to    
$e_2$, such that $e_2$ is at least as good as $e_1$ \wrt the chosen measure in all contexts. 
Our correctness notion is contextual equivalence, which means that $e_1, e_2$ behave identically \wrt termination in all contexts. 

Our approach is to use an abstract Sestoft-machine (see also \cite{gustavssonSands99}) as interpreter in order to have a realistic model for the 
resource consumption at runtime. 
Since space is an issue, in particular the maximally used space during an evaluation,  a detection of dynamically generated garbage is required, 
which leads to the implementation of an (eager) garbage collector. This is nontrivial, since letrec permits cyclic references. 
Furthermore,  indirections $x = y$ in \tletrec-environments turned out to lead to more space consumption
in the Sestoft machine     
than in the calculus,  
which is defeated by  removing  indirections at compile time as well as adapting the abstract machine.
Removing indirections can be done efficiently (see Section \ref{para-indirections}). 
The interpreter implementation is also shown to exactly count the maximal space usage for machine expressions (see Theorem \ref{thm:adequacy}).
Our  specific space analyses can exhibit in examples the  reason for unexpected
space increases, can refute transformations being space improvements, and can also give hints on the complexity of 
evaluations.

{\bf Outline }
In   Section ~\ref{sec:lrp} \LRP{} is introduced. In Section~\ref{sec:lrpgc} \LRP{} is extended with garbage collection.
Space improvements are explained in Section~\ref{sec:spimp}. In Section~\ref{sec:lrpint}  the \LRP{}-interpreter 
is described. The analyses and some results are in Section~\ref{sec:ana}. We conclude in Section~\ref{sec:conc}.

\section{Polymorphically Typed Lazy Lambda Calculus} \label{sec:lrp}
We recall  the Polymorphically Typed Lazy Lambda Calculus (\LRP) \cite{ppdp-IB55} as language. 
We also motivate and 
introduce 
several necessary extensions for supporting realistic space analyses.

\LRP{} \cite{schmidt-schauss-sabel-WPTE:14} is \LR{} (\eg see \cite{safenoecker}) extended with types. I.e., \LRP{} is an extension of the lambda calculus by polymorphic types, 
recursive \texttt{letrec}-expressions, \texttt{case}-expressions, \texttt{seq}-expressions, data constructors, type abstractions $\Lambda a.s$ 
to express polymorphic functions and  type applications $(s\ \tau)$ for type instantiations. The syntax of expressions and types of \LRP{} is defined in Fig.~\ref{fig:lrp}.

\begin{figure*}[!htb]
 \fbox{
\begin{minipage}{.98\textwidth}
\noindent{\bf Syntax of expressions and types:}
Let type variables $a,a_i\in\mathit{TVar}$ and term variables $x_,x_i\in\mathit{Var}$. 
Every type constructor $K$ has an arity $ar(K)\geq 0$ and a finite set $D_K$ of data constructors $c_{K,i}\in D_K$ 
with an arity $ar(c_{K,i})\geq 0$. \\[2mm]
\textbf{Types} $Typ$ and polymorphic types $\mathit{PTyp}$ are defined as follows:\\
$\begin{array}{lll}
\tau\in\mathit{Typ}  & ::= & a \mid (\tau_1 \rightarrow \tau_2) \mid (K\ \tau_1\ \dots\ \tau_{ar(K)})\\
\rho\in\mathit{PTyp} & ::= & \tau \mid \forall a.\rho
\end{array}$\\[2mm]
\textbf{Expressions} $\mathit{Expr}$ are generated by this grammar with $n\geq 1$ and $k\geq 0$:
\[\begin{array}{lll}
s,t\in\mathit{Expr} & ::= & u \mid x :: \rho \mid (s\ \tau) \mid (s\ t) \mid (\seq s\ t) \mid (\letrec x_1::\rho_1=s_1,\dots,x_n::\rho_n=s_n\lin t)\\
                    &     & \hspace*{-5mm} \mid (c_{K,i}::(\tau)\ s_1\ \dots\ s_{ar(c_{K,i})}) \mid (\case_K\ s\of\{(Pat_{K,1} \casepf t_1)\ \dots\ (Pat_{K,|D_K|} \casepf t_{|D_K|})\})\\
Pat_{K,i}           & ::= & (c_{K,i}::(\tau)\ (x_1::\tau_1)\ \dots\ (x_{ar(c_{K,i})}::\tau_{ar(c_{K,i})}))\\
u\in\mathit{PExpr}  & ::= & (\Lambda a_1.\Lambda a_2.\dots\Lambda a_k.\lambda x::\tau.s)
\end{array}\]

\end{minipage}}
\caption{Syntax of expressions and types\label{fig:lrp}}
\end{figure*}

An expression is {\em well-typed} if it can be typed using typing rules that are defined in \cite{schmidt-schauss-sabel-WPTE:14}. 
\LRP{} is a core language of Haskell and is simplified compared to Haskell, because it does not have type classes and is only polymorphic in the bindings of \letrec variables. 
But \LRP{} is strong enough to express polymorphically typed lists, and functions working on such data structures. 

From now on we use $Env$ as abbreviation for a \texttt{letrec}-environment, $\{x_{g(i)}=s_{f(i)}\}_{i=j}^m$ for $x_{g(j)}=s_{f(j)},\dots,x_{g(m)}=s_{f(m)}$ and $alts$ for
 \texttt{case}-alternatives. We use $FV(s)$ and $BV(s)$ to denote free and bound variables of an expression $s$ and $LV(Env)$ to denote the binding variables of a
  \texttt{letrec}-environment. Furthermore we abbreviate $(c_{K,i}\ s_1\ \dots\ s_{ar(c_{K,i})})$ with $c\ \vv{s}$ and $\lambda x_1.\dots\lambda x_n.s$ with 
  $\lambda x_1,\dots,x_n.s$. The data constructors \texttt{Nil} and \texttt{Cons} are used to represent lists, but we may also use the Haskell-notation [] and (:) instead.\newpage
A {\em context} $C$ is an expression with exactly one hole $[\cdot]$ at expression position. A {\em value} is an abstraction $\lambda x.s$, 
a type abstraction $u$ or a constructor application $c\ \vv{s}$. 

After the reduction position is determined using the labeling algorithm of \cite{schmidt-schauss-sabel-WPTE:14}, 
a unique reduction rule of Fig.~\ref{fig:basred} is applied at this position
which constitutes a normal-order reduction step.

\begin{figure*}[!htb]
 \fbox{
\begin{minipage}{.98\textwidth}
$\begin{array}{@{}ll@{}}
\text{(lbeta)}   & C[((\lambda x.s)^\lsub\ r)] \rightarrow C[(\letrec x=r\lin s)]\\
\text{(Tbeta)}   & ((\Lambda a.u)^\lsub\ \tau) \rightarrow u[\tau/a]\\
\text{(cp-in)}   & (\letrec x_1=v^\lsub,\{x_i=x_{i-1}\}_{i=2}^m,Env\lin C[x_m^\lvis])\\
                 & \rightarrow (\letrec x_1=v,\{x_i=x_{i-1}\}_{i=2}^m,Env\lin C[v])\\
                 & \quad \text{where }v\text{ is a polymorphic abstraction}\\
\text{(cp-e)}    & (\letrec x_1=v^\lsub,\{x_i=x_{i-1}\}_{i=2}^m,Env,y=C[x_m^\lvis]\lin r)\\
                 & \rightarrow (\letrec x_1=v,\{x_i=x_{i-1}\}_{i=2}^m,Env,y=C[v]\lin r)\\
                 & \quad \text{where }v\text{ is a polymorphic abstraction}\\
\text{(llet-in)} & (\letrec Env_1\lin (\letrec Env_2\lin r)^\lsub) \rightarrow (\letrec Env_1,Env_2\lin r)\\
\text{(llet-e)}  & (\letrec Env_1,x=(\letrec Env_2\lin t)^\lsub\lin r) \rightarrow (\letrec Env_1,Env_2,x=t\lin r)\\
\text{(lapp)}    & C[((\letrec Env\lin t)^\lsub\ s)] \rightarrow C[(\letrec Env\lin (t\ s))]\\
\text{(lcase)}   & C[(\texttt{case}_K\ (\letrec Env\lin t)^\lsub\of alts)] \rightarrow C[(\letrec Env\lin (\texttt{case}_K\ t\of alts))]\\
\text{(seq-c)}   & C[(\seq v^\lsub\ t)] \rightarrow C[t]\qquad \text{if }v\text{ is a value}\\
\text{(seq-in)}  & (\letrec x_1=(c\ \vv{s})^\lsub,\{x_i=x_{i-1}\}_{i=2}^m,Env\lin C[(\seq x_{m}^\lvis\ t)])\\
                 & \rightarrow (\letrec x_1=v,\{x_i=x_{i-1}\}_{i=2}^m,Env\lin C[t]) \qquad \text{if }v\text{ is a value}\\
\text{(seq-e)}   & (\letrec x_1=(c\ \vv{s})^\lsub,\{x_i=x_{i-1}\}_{i=2}^m,Env,y=C[(\seq x_{m}^\lvis\ t)]\lin r)\\
                 & \rightarrow (\letrec x_1=v,\{x_i=x_{i-1}\}_{i=2}^m,Env,y=C[t]\lin r)\qquad \text{if }v\text{ is a value}\\
\text{(lseq)}    & C[(\seq (\letrec Env\lin s)^\lsub\ t)] \rightarrow C[(\letrec Env\lin (\seq s\ t))]\\
\text{(case-c)}  & C[(\case_K\ c^\lsub\of \{ \dots (c\rightarrow t) \dots\} )] \rightarrow C[t] \quad \text{if }ar(c)=0\text{, otherwise:}\\
                  & C[(\case_K\ (c\ \vv{x})^\lsub\of \{ \dots ((c\ \vv{y})\rightarrow t) \dots \})] \rightarrow C[(\letrec \{y_i=x_i\}_{i=1}^{ar(c)}\lin t)]\\
\text{(case-in)} & (\letrec x_1=c^\lsub, \{x_i=x_{i-1}\}_{i=2}^m,Env\lin C[(\case_K\ x_m^\lvis\of \{\dots(c\rightarrow r)\dots\})])\\
                  & \enskip \rightarrow (\letrec x_1=c, \{x_i=x_{i-1}\}_{i=2}^m,Env\lin C[r])\qquad \text{ if }ar(c)=0;~\text{otherwise:}\\
                  & (\letrec x_1=(c\ \vv{t})^\lsub,\{x_i=x_{i-1}\}_{i=2}^m,Env\lin C[(\case_K\ x_m^\lvis\of \{ \dots ((c\ \vv{z})\rightarrow r) \dots \})])\\
                 &  \rightarrow (\letrec x_1=(c\ \vv{y}),\{y_i=t_i\}_{i=1}^{ar(c)}, \{x_i=x_{i-1}\}_{i=2}^m,Env\lin C[\letrec \{z_i{=}y_i\}_{i=1}^{ar(c)}\lin r])\\
\text{(case-e)}  & (\letrec x_1=c^\lsub, \{x_i=x_{i-1}\}_{i=2}^m, u=C[(\case_K\ x_m^\lvis\of \{\dots(c\rightarrow r_1)\dots\})],\quad Env\\
                 & \enskip\lin r_2)\\
                  & \enskip \rightarrow (\letrec x_1=c, \{x_i=x_{i-1}\}_{i=2}^m, u=C[r_1],Env\lin r_2) \qquad \text{ if }ar(c)=0;~\text{otherwise:}\\
                  & (\letrec x_1=(c\ \vv{t})^\lsub,\{x_i=x_{i-1}\}_{i=2}^m,\\
                 & \enskip\quad u=C[(\case_K\ x_m^\lvis\of \{ \dots((c\ \vv{z}) \rightarrow r)\dots \})],Env\lin s)\\
                 & \enskip \rightarrow (\letrec x_1=(c\ \vv{y}), \{y_i=t_i\}_{i=1}^{ar(c)}, \{x_i=x_{i-1}\}_{i=2}^m,\\[1mm]
                 & \enskip\qquad u=C[\letrec \{z_i=y_i\}_{i=1}^{ar(c)}\lin r],Env\lin s)\\
\end{array}$
\end{minipage}}
\caption{Basic reduction rules. The variables $y_i$ are fresh.\label{fig:basred}}
\end{figure*}

The classical $\beta$-reduction is replaced by the sharing (lbeta). 
(Tbeta) is used for type instantiations concerning polymorphic type bindings. The rules  (cp-in) and (cp-e) copy abstractions which is needed 
when the reduction rule has to reduce an application $(f\ g)$ where $f$ is an abstraction defined in a \texttt{letrec}-environment. 
The rules (llet-in) and (llet-e) are used to   merge nested \texttt{letrec}-expressions;
 (lapp), (lcase) and (lseq) move a \texttt{letrec}-expression out of an application, a \texttt{seq}-expression or a \texttt{case}-expression;
  (seq-c), (seq-in) and (seq-e) evaluate \texttt{seq}-expressions, where the first argument has to be a value or a value which is reachable through 
  a \texttt{letrec}-environment.
  (case-c), (case-in) and (case-e) evaluate 
  \texttt{case}-expressions by using \texttt{letrec}-expressions to realize the insertion of the variables for the appropriate
  \texttt{case}-alternative. 

The following abbreviations are used:  (cp) is the union of (cp-in) and (cp-e);
 (llet) is the union of (llet-in) and (llet-e);
  (lll) is the union of (lapp), (lcase), (lseq) and (llet);
  (case) is the union of (case-c), (case-in), (case-e);
  (seq) is the union of (seq-c), (seq-in), (seq-e).

Normal order reduction steps and notions for termination are defined as follows:
\begin{definition}[Normal order reduction] \label{def:redstep}
A {\em normal order reduction step} $s\rnoLRP t$ is performed (uniquely) if the labeling algorithm in \cite{schmidt-schauss-sabel-WPTE:14} terminates 
 on $s$, inserting $\lsub$ (subexpression) and $\lvis$  (visited by the labeling), and the applicable rule of
 Fig.~\ref{fig:basred} produces $t$. The notation
 $\rnortLRP$  is the reflexive, transitive closure, 
  $\rnopLRP$ is the  transitive closure of $s\rnoLRP t$; and 
  $\rnokLRP$ denotes $k$ normal order steps.
\end{definition}

\begin{definition} \label{def:whnf}
\begin{compactenum}
\item A weak head normal form (WHNF) is a value, or an expression $\letrec Env\lin v$, where $v$ is a value, or an expression $\letrec x_1=c\ \vv{t},\{x_i=x_{i-1}\}_{i=2}^m,Env\lin x_m$.
\item An expression $s$ {\em converges} to an expression $t$ ($s\conv t$ or $s\conv$ if we do not need $t$) if $s\rnortLRP t$ where $t$ is a WHNF.
  Expression $s$ {\em diverges} ($s\diverges$) if it does not converge. 
\item The symbol $\bot$ represents a closed diverging expression, \eg $\letrec~x=x~\lin~x$. 
\end{compactenum}
\end{definition}

\begin{definition} \label{def:equi}
For \LRP-expressions $s,t$, $ s\leq_c t$ holds iff $\forall C[\cdot]:C[s]\conv \Rightarrow C[t]\conv$,
and $s\sim_c t$ holds iff $s\leq_ct$ and $t\leq_cs$. The relation $\leq_c$ is called {\em contextual preorder} and $\sim_c$ is called {\em contextual equivalence}.
\end{definition}

The following notion of reduction length is used for measuring the time behavior in \LRP{}.
\begin{definition} \label{def:rln}
For a closed \LRP-expression $s$ with $s\conv s_0$, let $\rln(s)$ be the sum of all (lbeta)-, (case)- and (seq)-reduction steps in $s\conv s_0$,
and let   $\rlnall(s)$ be the number of all reductions, but not (TBeta), in $s\conv s_0$.
\end{definition}

\section{\LRP{} with Eager Garbage Collection} \label{sec:lrpgc}
The calculus \LRP{} does not remove garbage itself.
However, for measuring the space-behavior, garbage should be ignored (and removed). Thus in this section
we add reduction rules for removing garbage, and show that an evaluation strategy with
garbage collection does not change the semantics of the calculus. In Fig.~\ref{fig:gcred} 
the rules for garbage collection are defined. We use (gc) for the union of (gc1) and (gc2).
\begin{figure*}[!htb]
 \fbox{
\begin{minipage}{.98\textwidth}
$\begin{array}{ll}
\text{(gc1)}     & (\letrec\{x_i=s_i\}_{i=1}^n,Env\lin t) \rightarrow (\letrec Env\lin t) \quad\text{if for all }i:x_i\notin FV(t,Env)\\
\text{(gc2)}     & (\letrec x_1=s_1,\ \dots,\ x_n=s_n\lin t) \rightarrow t  \quad\text{if for all }i:x_i\notin FV(t)\\
\end{array}$

\end{minipage}}
\caption{Garbage collection rules\label{fig:gcred}}
\end{figure*}

Since we focus on space improvements,      
it is useful to model eager garbage collection also in the calculus, which leads to the calculus \LRPgc{}.
It collects (dynamic) garbage only in the top \tletrec, which is sufficient to remove all (reference-) garbage, if the starting
program does not contain garbage.

\begin{definition} \label{def:nogc}\label{def:lrpgc}
\LRPgc{} is \LRP{} where the normal-order reduction is modified as follows:\\ 
Let $s$ be an \LRP{}-expression.    A {\em normal-order-gc (nogc) reduction step} is defined by two cases:\\[-3mm]
\begin{compactenum}
\item If a (gc)-transformation is applicable to $s$ in the top \tletrec, then this transformation is applied to $s$, where the maximal number
 of bindings is removed.
\item If 1. is not applicable and an \LRP{}-normal-order reduction step is applicable to $s$, then this normal-order reduction is applied to $s$.
\end{compactenum}\ \\[-3mm]
A sequence of nogc-reduction steps is called an {\em nogc-reduction sequence}. 
An \LRPgc{}-WHNF $s$ is either an \LRP-WHNF which is not a {\tletrec} expression, or it is an \LRP-WHNF that is a  \tletrec-expression which does not permit (gc)-transformation in 
the top {\tletrec}. 
If for $s$, there is an nogc-reduction sequence that leads to an \LRP{gc}-WHNF,  then we say $s$ {\em converges \wrt \LRPgc{}} and write $s\convgc$. \\
In $\LRP{gc}$, the equivalence $s \sim_{c,\nogc} t$ is defined as for $\LRP$, but \wrt $\convgc$.
\end{definition}
Several subsequent (gc)-reductions are possible in an $\LRP{gc}$-normal-order reduction sequence, for example a  (gc2)-reduction followed by a (gc1)-reduction.

We will show in the following that the calculi $\LRP$ and $\LRPgc$  are equivalent \wrt convergences as well as \wrt the rln-measure. \\

In the following we will use complete sets of forking (and commuting) diagrams (more information on this technique is in \cite{safenoecker}).
 A forking is an overlapping between a normal-order transformation and a non-normal-order transformation 
(also called internal transformation). A complete set of forking diagrams for transformation $b$ contains a forking diagram for each possible 
forking of the form 
$s_2 \xleftarrow{nogc} s_1 \xrightarrow{b} s_1'$. 
The treatment is similar for commuting diagrams and commuting situations $s_1 \xrightarrow{b} s_1' \xrightarrow{nogc} s_2'$. 
We will use the notation $(nogc,a)$ which is an arbitrary nogc-reduction  if not otherwise stated. 
If the label $a$ is used twice,  then all occurrences of $a$ represent the same rule.
Let $\LCSC := \{\text{(lbeta)}, \text{(case)}, \text{(seq)}, \text{(cp)}\}$. %

\begin{lemma}\label{lemma-gc-diagrams}
The forking diagrams between a nogc-reduction and a non-normal-order (gc)-transforma\-tion in \LRP{gc} in any context are the following:
\[
\begin{array}{llll}
\begin{minipage}[t]{0.24\textwidth}
\xymatrix@R=6mm@C=12mm{
 s_1 \ar[r]^{gc} \ar[d]_{nogc,a} & s_1'\ar@{-->}[d]_{nogc,a} \\
 s_2 \ar@{-->}[r]^{gc}& s_2' 
}
\end{minipage}
& 
\begin{minipage}[t]{0.24\textwidth}
\xymatrix@R=6mm@C=12mm{
 s_1 \ar[rr]^{gc} \ar[d]_{nogc,cp} && s_1'\ar@{-->}[d]_{nogc,cp} \\
 s_2 \ar@{-->}[r]^{gc} &  s_3 \ar@{-->}[r]^{gc} & s_2'  
}
\end{minipage}
&
\begin{minipage}[t]{0.24\textwidth}
\xymatrix@R=6mm@C=12mm{
 s_1 \ar[r]^{gc} \ar[d]_{nogc,a} & s_1'\ar@{-->}[dl]_{nogc,a} \\
 s_2 
}
\end{minipage} 
&
\begin{minipage}[t]{0.24\textwidth}
\xymatrix@R=6mm@C=12mm{
 s_1 \ar[r]^{gc} \ar[d]_{nogc,lll} & s_1'\\
 s_2\ar@{-->}[ur]_{gc}  
}
\end{minipage}
\end{array}
\]
The commuting diagrams can be immediately derived from the forking diagrams. 
\end{lemma}
\begin{proof}
The first diagram occurs if the nogc-reduction and the transformation can be commuted. The second diagrams happens if the gc-transformation
was done in the copied abstraction. The third diagram occurs, if the effect of the gc-transformation was also done by the nogc-reduction,
where we assume that $\xrightarrow{gc}$ and $\xrightarrow{nogc,a}$ are different.
Finally the fourth diagram occurs for example in  
$(\letrec Env_1 \lin (s_1~s_2)) \xleftarrow{\nogc} ((\letrec Env_1\lin s_1)~s_2) \xrightarrow{gc}$ $(s_1~s_2)$ and where
$(\letrec  Env_1 \lin (s_1~s_2))$  $\xrightarrow{gc}$ $(s_1~s_2)$.  
\end{proof}

\begin{theorem}\label{thm:equivalence}
\LRP{} and \LRP{gc} are convergence-equivalent, \ie{} for all expressions $s$: $s\conv \iff s\convgc$.
\end{theorem}
\begin{proof}
If $s\convgc$ then $s\conv$ holds, since (gc) and all reductions of the calculus are  correct \wrt \LRP-normal-order reduction,
  which follows from their untyped correctness (see \cite{safenoecker}).\\
Under the assumption that (gc) is correct in \LRP{gc}, it is straightforward to show that $s\conv$ implies $s\convgc$. \\
It remains to show that (gc) is correct in \LRP{gc}:
Therefore we have to use the diagrams in  Lemma \ref{lemma-gc-diagrams} for (gc). 
We consider the situation $s_0 \xleftarrow{\nogc,*}  s_1 \xrightarrow{gc}$ $s_1'$ 
where $s_0$ is an $\LRPgc$-WHNF. For the induction proof we consider the smaller diagram
  $s_2 \xleftarrow{\nogc}  s_1 \xrightarrow{gc}$ $s_1'$ and show that there is a nogc-reduction of $s_1'$ such that $\rlnall(s_1') \leq \rlnall(s_1)$. 
  First we observe that $\LRPgc$-WHNFs remain $\LRPgc$-WHNFs under (gc). \\
  The induction measure is $\rlnall(s_1)$. 
  For the situation $s_2 = s_1'$ or if     
  any of the four diagrams applies to the situation,  
  the induction hypothesis applies, where in case of diagrams 2, we have to apply it twice.  This shows that there is a nogc-reduction  of  $s_1$ to a 
  $\LRPgc$-WHNF.
  
  The second part is to consider the situation $s_1 \xrightarrow{gc}$ $s_1'\xrightarrow{\nogc,*} s_0'$, where $s_0'$ is an $\LRPgc$-WHNF. 
  Here we show more: that there is an nogc-reduction of $s_1$ with $\rln_{LCSC}(s_1) \leq \rln_{LCSC}(s_1')$, 
  where  $\rln_{LCSC}$ counts the normal-order reductions from $\LCSC$ until a WHNF is reached.
    The induction is on the lexicographic combination of the measures $(\rln_{LCSC}(s_1'),$ $\mu_{lll}(s_1),$ $|s_1'|,\rlnall(s_1'))$, where 
    $\mu_{lll}$ is the 
    measure from \cite{safenoecker} that is strictly decreased by every $\xrightarrow{lll}$ and $\xrightarrow{gc}$-reduction, and $|s_1'|$ is the 
    size of $s_1'$ as an expression. 
    If $s_1'$ is an \LRP-WHNF, then either $s_1 \xrightarrow{gc} s_1'$ is a normal-order reduction, and we are done, 
    or it is
     not a  normal-order reduction, and $s_1$ is also an \LRP-WHNF. 
    
    If  $s_1 \xrightarrow{gc}$ $s_1'$ is an nogc-reduction, then the claim holds. 
    In the case of the first diagram, the induction hypothesis can be applied by the following reasoning:
    if the $s_1'$-reduction is a LCSC-reduction, then the measure is decreased; if it is an (lll) or (gc), then the 
    first component  is the same but pair of the second and third component is strictly smaller.  In the case of the second diagram, $\rln_{LCSC}(s_2')$ is strictly
    smaller, and hence also, by the induction hypothesis,  $\rln_{LCSC}(s_3)$  and we can again apply the induction hypothesis.
    In the case of the third diagram, reasoning is obvious. Finally, in the case of the fourth diagram, $\mu_{lll}(s_2) < \mu_{lll}(s_1')$, hence the 
    induction hypothesis can be applied. 
      \qedhere 
\end{proof}

\begin{corollary}
The contextual equivalences of $\LRP$  and $\LRP{gc}$  are identical.
\end{corollary}
The proof of Theorem \ref{thm:equivalence} also shows that the rln-measure of expressions is the same for $\LRP$ and $\LRP{gc}$.  
Hence we can drop the distinction between {\LRP} and \LRP{gc} \wrt $\rln$ as well as for $\sim_c$.

\section{Time- and Space-Improvements} \label{sec:spimp}
For space analyses, we first define the size of expressions:
\begin{definition} \label{def:size}
The size $\size(s)$ of an expression $s$ is the following number:\\
$
\begin{array}{@{}c@{~~}c@{}}
\begin{array}{@{}l@{~}c@{~}l@{}}
\size(x)                                          &=& 0\\
\size(c\vv{s})                       &=& 1+\sum_{i=1}^n\size(s_i)\\
\size(\lambda x.s)                              &=& 1+\size(s)\\
\size(c\vv{x} \casepf e)              &=& 1+\size(e)\\
\end{array}
&
\begin{array}{@{}l@{~}c@{~}l@{}}
\size(s\ t)                                     &=& 1+\size(s)+\size(t)\\
\size(\seq s_1\ s_2)                            &=& 1+\size(s_1)+\size(s_2)\\
\size(\letrec \{x_i=s_i\}_{i=1}^n \lin s) &=& \size(s)+\sum_{i=1}^n\size(s_i)\\
\size(\case\ e\of \{alt_1\ \dots\ alt_n\})      &=& 1+\size(e)+\sum_{i=1}^n\size(alt_i)\\
\end{array}
\end{array}
$
\end{definition}

Type annotations are not counted by the size measure and thus they are also not shown in the definition of $\size$.
Note that our chosen size measure also does not count variables, the number of \texttt{letrec}-bindings, nor the \texttt{letrec}-label itself.
This can be justified, since these constructs are usually represented more efficiently (or do not occur) in realistic implementations,
for example in the abstract machine. 

For measuring the space-behavior of $s$, we use the maximum size occurring in an $nogc$-reduction sequence to a WHNF:
\begin{definition} \label{def:spmax}
Let $s$ be a closed \LRP{}-expression. If $s=s_0\rnogc s_1\rnogc \dots \rnogc s_n$ where $s_n$ is a WHNF, then $\spmax(s)$ is the maximum of $\size(s_i)$. 
If $s\diverges$ then $\spmax(s)=\infty$.
\end{definition}
This measure is very strict and especially appropriate if the available space is limited. 
A transformation is a time improvement \cite{schmidt-schauss-sabel-PPDP:2015,ppdp-IB55} if it never increases the $\rln$-reduction length,
and a transformation is a space improvement if it never increases the space consumption. 
\begin{definition} \label{def:spmax-improvement}
Let $s,t$ be two expressions with $s\sim_c t$.  
Then $s$ is a {\em maxspace-improvement} of $t$, $s\spmaxleq t$, if for all contexts $C$: If $C[s]$, $C[t]$ are closed then $\spmax(C[s]) \leq  \spmax(C[t])$.\\
We say $s$ {\em (time-)improves} $t$, $s \preceq t$, if for all contexts $C$: If $C[s]$, $C[t]$ are closed, then $\rln(C[s]) \leq  \rln(C[t])$.
\end{definition}

These relations are precongruences. Note that we use $n < \infty$, and $\infty \leq \infty$.

\section{An Abstract Machine for \LRP} \label{sec:lrpint}

In this section we present the abstract machine (a variant of the Sestoft-machine) to evaluate \LRP-programs and measure their time and space usage.
However, the conceptually simple abstract machine has to be extended and adapted to obtain a good  behavior w.r.t. space measuring: 
it must be able to remove unused bindings in 
{\tletrec}s, and it has to prevent superfluous duplications of expressions    
in the input as well as their dynamic creation.  
A first step is to transform the \LRP-expressions into so-called machine expressions 
on which the Sestoft-machine can be applied. These are \LRP-expressions with the restriction that 
arguments of applications, constructor applications, and the second argument of {\tt seq} must be variables.
We also remove all type information.   
\begin{definition} \label{def:mtrans}
The translation $\psi$ from arbitrary \LRP{}-expressions into machine expressions is defined as follows,
 where $y,y_i$ are fresh variables:
$$
\begin{array}{l@{\qquad\qquad}l}
\begin{array}[t]{@{}l@{~}c@{~}l@{}}
  \psi(x :: \rho) &\hspace*{0.6mm}:= &x
\\
  \psi(s\ \tau) &\hspace*{0.6mm}:= &\psi(s)
\\
  \psi(\Lambda a_1.\Lambda a_2.\dots.\Lambda a_k.\lambda x{::}\tau.s) &\hspace*{0.6mm}:=& \lambda x.\psi(s)
\\
\end{array}
&
\begin{array}[t]{@{}l@{~}c@{~}l@{}}
  \psi(s\ t)    &:= &\letrec y=\psi(t)\lin (\psi(s)\ y)
\\
  \psi(\seq s\ t) &:= &\letrec y=\psi(t)\lin (\seq \psi(s)\ y)
\\
  \psi(c~\vv{s}) &:=& \begin{array}[t]{@{}l@{}}\letrec \{y_i=\psi(s_i)\}_{i=1}^n\lin (c~\vv{y_i})\end{array}
\\

\end{array}
\\
\multicolumn{2}{l@{}}{
\begin{array}{@{}l@{~}c@{~}l@{}}
  \psi(\letrec~\{x_i=s_i\}_{i=1}^n \lin t) := \letrec \{x_i=\psi(s_i)\}_{i=1}^n \lin \psi(t)
\\
  \psi(\case_K\ e\of \{ (Pat_{K,1} \casepf t_1)\ \dots\ (Pat_{K,|D_K|} \casepf t_{|D_K|})\})\\
 \multicolumn{3}{r}{\hspace*{4.45cm}:= \case_K\ \psi(e)\of \{ (Pat_{K,1} \casepf \psi(t_1))\dots(Pat_{K,|D_K|} \casepf \psi(t_{|D_K|}))\}}
\\
\end{array}
}
\end{array}
 $$
\end{definition}
The transformation adds \texttt{letrec}-expressions and removes type annotations. 
This transformation does not change the reduction length, \ie $\rln(s)=\rln(\psi(s))$ (see \cite{ppdp-IB55}). 
It is easy to see that $\size(s)=\size(\psi(s))$ holds.
Below we will show that for machine expressions $s$ the value $\spmax(s)$ is correctly computed. Unfortunately, this does not hold in general:
for example $((\tseq~\True~(\lambda x.a))~\True)$ and $(\tletrec~x_1 = \True,x_2 =  \lambda x.a~\tin~(\tseq~\True~x_2)~x_1)$ have different 
$\spmax$-values for $\size(a) \ge 1$: 
$5+\size(a)$ and $4+2\size(a)$, respectively, since the latter has a space peak at 
$(\tletrec~x_1 = \True,x_2 = \lambda x.a~\tin~(\lambda x.a)~x_1)$.  
 
The used abstract machine is defined in \cite{schmidt-schauss-sabel-PPDP:2015,ppdp-IB55} and is based on the abstract machine \absm{} by Peter Sestoft (see \cite{sestoft}),
which was designed for call-by-need evaluation. 
The machine is extended in a straightforward way to handle \texttt{seq}-expressions, where 
a \texttt{seq}-expression evaluates the first argument to a value and then returns the second argument.
A state $Q$ is a triple $\langle \Gamma \mid s \mid S\rangle$, where $\Gamma$ is an environment of variable-to-expression bindings (sometimes called heap), 
$s$ is a machine expression (often called control expression) and $S$ is a stack with entries $\RetApp(x)$, $\RetSeq(x)$, $\RetCase(alts)$ and $\RetHeap(x)$ where $x$ is a 
variable and $alts$ is a list of case alternatives.
 Because the stack is implemented as a list we sometimes use the usual list notation for the stack.
The control expression is the expression which has to be evaluated next, together with the stack it controls the control flow of the program. 
The stack is also responsible to trigger updates on the heap. 
Note that the WHNFs of the abstract machine are machine expressions that are WHNFs.

The abstract machine is defined in Fig.~\ref{fig:mark1}. The execution of a program starts with the whole program as control expression and an empty heap and stack. 
The transition rules define the transition from one state to the next, where at most one rule is applicable in each step. 

        The rules (Unwind1), (Unwind2), and (Unwind3) perform the search for the redex (according to the labeling in $\LRP$), by storing arguments of applications, 
        \texttt{seq}-expressions, or \texttt{case}-alternatives on the stack.
        The rule 
        (Lookup) moves heap bindings into the scope of evaluation (if they are demanded). 
If evaluation of a binding is finished, the rule (Update) restores the result in the heap. 
(Letrec) moves \texttt{letrec}-bindings into the heap, by creating new heap bindings.
        (Subst) is applicable if the first argument of an application is evaluated to an abstraction and the stack contains the argument. It then performs
        a $\beta$-reduction (with a variable as argument).
(Branch) analogously performs a (case)-reduction on the abstract machine.
(Seq) evaluates a \texttt{seq}-expression.
Rule (Blackhole) results in an infinite loop, i.e. an error.

The abstract machine iteratively applies these rules until a final state is reached. 
Note that the control expression of a state is a \absm{} value if no rule is applicable.

The  (optional) rule (GC)  performs garbage collection of bindings. The (optional) rule (SCRem) performs a specific form of   saving space:  
it prevents unnecessary copying of values by  
 avoiding the intermediate  construction of indirections $y = x$  and applying  the replacement instead.
For a correct space measurement, these rules have to be applied whenever possible.

\begin{figure*}[!htb]
 \fbox{
\begin{minipage}{.97\textwidth}   
\noindent{\bf Initial state:}  \ $\langle\emptyset \mid e \mid []\rangle$ where $e$ is a machine expression.

\noindent{\bf Transition rules:}\\[.3ex]
$\begin{array}{ll}
\text{(Unwind1)}   & \langle \Gamma \mid (s\ x) \mid S \rangle \rightarrow \langle \Gamma \mid s \mid \RetApp(x):S \rangle\\
\text{(Unwind2)}   & \langle \Gamma \mid (\seq s\ x) \mid S \rangle \rightarrow \langle \Gamma \mid s \mid \RetSeq(x):S \rangle\\
\text{(Unwind3)}   & \langle \Gamma \mid \case_K\ s \of alts \mid S \rangle \rightarrow \langle \Gamma \mid s \mid \RetCase(alts):S \rangle\\
\text{(Lookup)}    & \langle \Gamma, x=s \mid x \mid S\rangle \rightarrow \langle \Gamma \mid s \mid \RetHeap(x):S\rangle\\
\text{(Letrec)}    & \langle \Gamma \mid \letrec Env \lin s \mid S\rangle \rightarrow \langle \Gamma,Env \mid s \mid S \rangle\\
\text{(Subst)}     & \langle \Gamma \mid \lambda x.s \mid \RetApp(y):S \rangle \rightarrow \langle \Gamma \mid s[y/x] \mid S \rangle\\
\text{(Branch)}    & \langle \Gamma \mid c_{K,i}\ \vv{x} \mid \RetCase(\dots\ ((c_{K,i}\ \vv{y}) \casepf t)\ \dots):S \rangle \rightarrow \langle \Gamma \mid t[\vv{x}/\vv{y}] \mid S \rangle\\
\text{(Seq)}       & \langle \Gamma \mid v \mid \RetSeq(y):S \rangle \rightarrow \langle \Gamma \mid y \mid S\rangle \quad \text{if }v\text{ is a \absm{} value}\\
\text{(Update)}    & \langle \Gamma \mid v \mid \RetHeap(x):S\rangle \rightarrow \langle \Gamma, x=v \mid v \mid S \rangle \quad \text{if }v\text{ is a \absm{} value}\\
\text{(Blackhole)} & \langle \Gamma \mid y \mid S \rangle \rightarrow \langle \Gamma \mid y \mid S \rangle \quad \text{if no binding for }y\text{ exists on the heap}\\
\end{array}$
\\[.2ex]
\noindent{\bf Garbage Collection and Stack Chain Removal (both optional):}\\[.2ex]
$\begin{array}{@{}ll@{}}
\text{(GC)}        & \hspace*{-7mm} \langle \Gamma,\{x_i=s_i\} \mid s \mid S \rangle \rightarrow \langle \Gamma \mid s \mid S \rangle \quad \text{where }\{x_i=s_i\}\text{ is the maximal set such that for all i:}\\
  &  \hspace*{-7mm} x_i\notin FV(\Gamma), x_i\notin FV(s), \RetApp(x_i)\notin S, \RetSeq(x_i)\notin S,\text{ and if }x_i \in FV(alts)\text{ then }\RetCase(alts)\notin S\\
\text{(SCRem)} & \langle \Gamma \mid s\mid \RetHeap(x):\RetHeap(y):S \rangle \rightarrow \langle \Gamma[x/y] \mid s[x/y] \mid \RetHeap(x):S[x/y]\rangle
\end{array}$
\\[.2ex]
\noindent{\bf Value:} A machine expression is a {\em \absm{} value} if it is an abstraction or constructor application.

\noindent{\bf WHNF:} Let $v$ be a \absm{} value. Then a machine expression is a {\em \absm{}-WHNF} if it is a \absm{} value or of the form $\letrec x_1 = e_1,\ \dots,\ x_n = e_n \lin v$.

\noindent{\bf Final State:} Let $v$ be a \absm{} value, then a final state is: \ $\langle \Gamma \mid v \mid []\rangle$
\end{minipage}}
\caption{Mark1: Initial state,  transition rules, value, WHNF and final state\label{fig:mark1}}
\end{figure*}

The rule (Update) is only applicable if (Lookup) was used before, hence (Letrec) is the only rule which is able to add completely new bindings to the heap. 

Moreover every (Lookup) triggers an (Update). 
There are situations where a variable as control expression leads to another variable as control expression (\eg variable chains in \texttt{letrec}-environments).
For example the state $\langle \Gamma \mid \True \mid \RetHeap(x):\RetHeap(y):\RetHeap(z):S\rangle$ leads to three (Update) in sequence. 
Seen as a \texttt{letrec}-environment, ${\tletrec}~ x = y, y = z, z = \True$ leads to ${\tletrec}~ x = \True, y = \True, z = \True$. But if we consider the rules in 
Fig.~\ref{fig:basred}, we see that \LRP{} does copy such values right to the needed position, without copying it to each position of the corresponding chain. 
The following example even shows that the difference in space consumption is at least $c\cdot n$, where $c$ is the size of the value $v$:
\[\letrec id = (\lambda x.x), x_1 = (id\ x_2), \dots, x_{n-1} = (id\ x_n), x_n = v \lin \texttt{seq}\ x_1\ (\texttt{T}\ x_1\ x_2\ \dots\ x_n)\]
The tuple $(\texttt{T}\ x_1\ x_2\ \dots\ x_n)$ ensures that none of the bindings can be removed by the garbage collector. 
Machine execution leads to a sequence of $n$ (Update)-transitions, where the value $v$ gets copied to each binding of the chain.
To avoid this effect,
the rule (SCRem) has to be applied whenever possible. If we consider the example above, then we have:
\[\langle \Gamma \mid \True \mid \RetHeap(x):\RetHeap(y):\RetHeap(z):S\rangle \xrightarrow{\text{(SCRem),2}} \langle \Gamma[x/y,x/z] \mid \True \mid \RetHeap(x):S[x/y,x/z]\rangle\]
The rule (SCRem) is correct, 
since $\langle \Gamma \mid v \mid \RetHeap(x):\RetHeap(y)\rangle$ corresponds to ${\tletrec}~\Gamma, x = v, y = x\ \tin\ y$ with $x \not= y$ before application, and 
after the application it is $\langle \Gamma[x/y] \mid v[y/x] \mid \RetHeap(x)\rangle$ corresponding to ${\tletrec}~\Gamma[x/y], x = v[y/x]\ \tin\ y[y/x]$ and replacing  
variables by variables 
is shown to be correct in \cite{safenoecker}.

Now we compare \LRP{} with the abstract machine:
\begin{definition} \label{def:mln}
Let $s$ be a closed machine expression such that $\langle \emptyset \mid s \mid [] \rangle \stackrel{n}{\rightarrow} Q$ where $Q$ is a final state.
\begin{enumerate}
\item $\mln(s)$ is the number of all (Subst)-, (Branch)- and (Seq)-steps in the sequence. 
\item $\mlnall(s)$ is the number of all machine steps in the sequence, thus $\mlnall(s)=n$. 
\item $\mathit{mspmax}(s)$ is $\max \{\size(St_i) \mid 1 \leq i \leq n, \neg (\mathit{St}_{i-1} = \langle\Gamma,c\vv{x},S\rangle \wedge 
\mathit{St}_{i-1} \xrightarrow{\mathrm{Update}} \mathit{St}_{i})\}$, (\ie states after (Update) are ignored for constructor applications), 
  where $\langle \emptyset \mid s \mid [] \rangle = St_1 \to St_2 \to \ldots  \to St_n = Q$.
\end{enumerate}
   If $s$ diverges then $\mln(s)=\mlnall(s)=\mathit{mspmax}(s) := \infty$.
\end{definition}

The size of a machine state is the sum of the heap sizes seen as outer \texttt{letrec}, the size of the control expression and the expressions on the stack, where the $\#$-labels are not counted. 
(Update) might increase the size of 
the current state in contrast to \LRP{}, where variables can be processed directly without looking up and then updating them (\eg compare (case-in) of \LRP{} with (Branch) 
of the Mark 1).

We show that the abstract machine can be used for computing reduction lengths and space measures
as needed for reasoning on time- and space-improvements (restricted to machine expressions in the case of space-improvements):  
\begin{theorem}[Adequacy of the abstract machine \wrt resource consumption]\label{thm:adequacy}
Let $s$ be an $\LRP$ expression with $s\conv$. 
\begin{enumerate}
  \item On input $\psi(s)$, the measure $\mln(\psi(s))$ coincides with $\rln(s)$.
\item If $s$ is a machine expression and if the abstract machine eagerly applies (GC) and (SCRem), then $\mathit{mspmax}(s)$ coincides with $\spmax(s)$.
\end{enumerate}
\end{theorem}
\begin{proof}
Since in \cite{ppdp-IB55} it was shown that  $\rln(s)=\mln(\psi(s))$ holds,    
\LRP{}, restricted to machine expressions, 
and Mark-1 provide equal results concerning reduction lengths. 
Note that this does not hold for $\rlnall$ and $\mlnall$, since the abstract machine moves \texttt{letrec}-environments directly on top, while \LRP{} needs additional 
(lll)-reduction steps.

Because bindings $x = y$ are eliminated by (SCRem) the only difference between evaluating the machine expressions $s$ in \LRPgc{} and the evaluation of $s$
 on the  abstract machine with 
eagerly applying rules (GC) and (SCRem) concerning space is the following:
The abstract machine copies constructor applications in contrast to \LRP{}. 
The constructor applications are either directly processed by a (Seq) or (Branch), or the copying is a final (Update)-transition.
The claim holds, since we do not count the sizes of exactly these intermediate states between (Update) and (Seq) as well as (Update) and (Branch), and a final (Update) 
in the computation of $\mathit{mspmax}(s)$.   
\end{proof}

\subsection{Implementation}
The \LRP{} interpreter (\LRPi) is implemented in Haskell and can be downloaded here:
\begin{center}\url{http://www.ki.informatik.uni-frankfurt.de/research/lrpi}\end{center}
All details concerning compilation can be found on this page.
The interpreter is able to execute \LRP{}-programs and to generate statistics concerning reduction lengths and different space measures. 
Various size and space measures can be defined easily, thus the interpreter can be  used to compare different size and space measurements or to explore other resource usages
 apart from time and space analyses.\\
The interpreter is user friendly and is able to calculate TikZ-pictures showing the \size{}-values during runtime (for use in LaTeX).

The rule (GC) is implemented as a stop-and-copy garbage collector that is called by the abstract machine depending on the 
garbage collection mode.
If we set the garbage collector to run after each state transition,
then the reduction length and $\spmax{}$-results (restricted to \LRP-machine-expressions in the case of space measurement) are correctly counted for \LRPgc{}, 
since the interpreter automatically applies (SCRem) whenever possible.  

\subsubsection{Removing Indirection Chains}\label{para-indirections}
We support the interpretation by two initial operations: 
There is a complete garbage collection before  starting the interpretation, and  an efficient algorithm to remove chains of indirections 
(variable-variable binding chains) in the input expression, 
which avoids unnecessary space consumption in the Sestoft machine. 
The algorithm is only applied once at compile time, since none of the rules in 
Fig.~\ref{fig:mark1} create variable-to-variable bindings that cannot be removed by (SCRem). 
Since we often configure the garbage collector to run whenever possible, this can reduce the runtime of garbage 
collection runs for large programs. 
This is implemented  efficiently and runs  in time $\mathcal{O}(n \log n)$ where $n$ is the number of variables.
For more information see \cite{dallmeyer:16}. 

\ignore{ \ldots
\paragraph{Removing Indirection Chains.}
We improve the garbage collector by an efficient algorithm to remove variable chains. The algorithm is only applied once at compile time, since none of the rules in 
Fig.~\ref{fig:mark1} create variable-to-variable bindings that cannot be removed by (SCRem). Since we often configure the garbage collector to run whenever possible, this can reduce the runtime of garbage 
collection runs for large programs. We describe and analyze this algorithm in the rest of this section.

To ensure termination, we first remove cyclic variable chains, which anyway lead to divergence if evaluated. Direct cycles can be removed using a scan.
The other cycles are removed using a graph that contains variable-to-variable bindings as edges, where all corresponding bindings of strongly connected components containing at least two nodes can be removed, since direct cycles are already removed.

In the second part we use a subalgorithm {\updateChain} which processes a single variable chain and yields a map that maps each variable of the chain to the last 
variable of the chain. This algorithm is called for each binding from the first to the last in this order and the map gets passed through. 
Thus, in the end we have a mapping from variables to the end point of their chains that is used to build the final \texttt{letrec}-environment.

\begin{algo} \label{algo:elim}
        \begin{compactenum}
        \item\label{algelim1} Scan all bindings and remove every variable-to-variable binding of the form $v=v$.
        \item\label{algelim2} Convert the remaining bindings to a directed graph $G(V,E)$ where $V$ are all variables that occur in variable-to-variable bindings and $(v,v')\in E$ holds for every variable-to-variable binding $v = v'$.
        \item\label{algelim3} Calculate the strongly connected components of $G$ 
   and remove every binding whose variables are part of a strongly connected component with more than one node. This provides a map $m$ with all leftover bindings, where the left-hand side of a binding is the key and the right-hand side is the value.
        \item\label{algelim4} Go through every binding from the first to the last: Let $b = e$ be the $k$-th binding. If $k=0$ then run the subalgorithm  {\updateChain} on that binding using $b$ 
            and $m$ as parameters, otherwise if $k>0$ then use $b$ and the map from the result of subalgorithm {\updateChain}  of the $k-1$-th binding as parameters for 
            {\updateChain}.
        \item\label{algelim5} Split the result of step \eqref{algelim4}    
             into a map only containing variable-to-non-variable bindings and a map only containing variable-to-variable bindings.
        \item\label{algelim6} If all bindings were removed, then return the \tin-expression of the \tletrec-expression.\\ 
          Otherwise the final \tletrec-expression is calculated using only the mappings of the map containing variable-to-non-variable bindings as bindings,
             while the map containing variable-to-variable bindings is used to make the needed substitutions at those binding-expressions and 
          the \tin-expression of the \tletrec.
        \end{compactenum}
        \ \\
  {\bf Subalgorithm} {\updateChain}  with variable $v$ and map $m$ as input: 
  \begin{quote}  Let $vars$ be an empty list. If the lookup of $v$ on $m$ returns a variable $v'$ then add $v$ to the end of $vars$ and run $chain$ with $v'$ and $m$. 
     Otherwise add key $k$ with value $v$ to $m$ for every $k$ in $vars$ and return this map.
   \end{quote}
  \end{algo}
\begin{proposition}\label{prop:chains}
Removal of variable-to-variable chains in an expression of (syntactical) size $n$ can be done in time $\mathcal{O}(n \log n)$. \end{proposition}
\begin{proof}
We show that our algorithm meets the complexity bound given in the claim:
For all steps but the fourth it is easy to see that these steps can be done in runtime complexity $\mathcal{O}(n\log n)$. 
The fourth step can also be done with this runtime complexity: Every variable chain is processed one after each other but in fact they are not processed completely separately, because the passed through map provides the end points of chains for each so far visited variable. Because cycles are already removed we never walk through a variable chain more than once. The only points where we do work twice is when we get to an already processed binding variable: We need logarithmic time to find out that this is the end of the current chain. Because we search for the end of a chain for each binding and we have $\mathcal{O}(n)$ bindings, this sums up to $\mathcal{O}(n\log n)$ steps only for the duplicated 
work at the end of chains. But apart from the end of chains we only evaluate each variable once, which needs logarithmic time for each variable and because there are $\mathcal{O}(n)$ variables we also need $\mathcal{O}(n\log n)$ steps here.
\end{proof}

\endignore }

\section{Analyses for Examples} \label{sec:ana}
This section contains analyses illustrating the performance and output of the interpreter and tool \LRPi. 
In particular it shows several experiments:
 a simple program transformation,  various fold-applications and a fusion,  comparing two list-reverse variants, 
 and sharing vs. non-sharing.  
The latter is illustrated by an example that can be seen as a variant of common subexpression elimination which shows that saving space may 
increase the runtime
and that a transformation, which for a large class of tests reduces the space, might fail to be a  space improvement in some cases. All analyses are done after translating the input to machine expression format.

One of the aims of \LRPi{} is to support conjectures of space improvements by affirmative tests, or to refute the space improvement property
of a specific transformation by finding a counter example. Since  \LRPi{} only tests in the empty environment, a complete test would require 
to perform the test also within contexts, which, however, cannot be done completely, since there are infinitely many, even using 
context lemmas to minimize the set of necessary contexts. Using a  simulation mode, the contexts could be restricted to testing the functions 
on arguments. For these tests typing makes a big difference, since certain transformations are correct only if typing is 
respected
and also the space improvement property may depend on the restriction to typed arguments or type-correct insertion into contexts. \\
Examples for conjectured space improvements are the reductions of the calculus (see Fig.\ref{fig:basred}) used as transformations, 
with the exception of the (cp)-reductions.

The used function definitions can be found in Fig.~\ref{fig:code}. 
The \texttt{fold}-function definitions are taken from \cite{ifl16}, \texttt{concat} and \texttt{concatMap} are inlined versions of the definitions 
in \cite{hasBase}. We first consider \texttt{fold}-functions of Haskell. \texttt{foldr} is the usual right-fold, \texttt{foldl} 
the usual left-fold and \texttt{foldl'} a more strict variant of \texttt{foldl}, which is not completely strict, since the used \texttt{seq} only 
evaluates $w$ until a value is achieved.
\begin{figure*}[!htb]
\begin{tabular}{@{}|@{}c@{\,}|@{}c@{}|}
\hline 
\begin{minipage}{.485\textwidth}
\[\begin{array}{@{\ }l@{~}l@{}}
\texttt{comp} &= \lambda f,g.(\lambda x.f\ (g\ x))\\
\foldr &= \lambda f,z,xs.\case\ xs \of \{ ([] \casepf z)\\
  & \quad\ \  ((y:ys) \casepf f\ y\ (\foldr\ f\ z\ ys)) \}\\
\foldl &= \lambda f,z,xs.\case\ xs \of \{ ([] \casepf z)\\
  & \quad\ \  ((y:ys) \casepf \foldl\ f\ (f\ z\ y)\ ys) \}\\
\foldls &= \lambda f,z,xs.\case\ xs \of \{ ([] \casepf z)\\
  & \quad\ \ ((y:ys) \casepf\\
  &\quad\ \ \quad\letrec w = (f\ z\ y)\\
  &\quad\ \ \ \quad\ \lin \seq w\ (\foldls\ f\ w\ ys))\}\\   
\texttt{map} &= \lambda f,lst.\case\ lst \of \{ ([] \casepf [])\\
    &\quad\; ((x:xs) \casepf ((f\ x):(\texttt{map}\ f\ xs)))\}\\
\texttt{tail} &= \lambda lst.\case\ lst\of \{\\
    &\quad\ \ \quad\:\ ([] \casepf \bot)\ ((x:xs) \casepf xs) \}\\
\texttt{replicate} &= \lambda n,x.\case\ n \of \{(\Zero \casepf [])\\
  &   ((\Succ\ m) \casepf x:(\texttt{replicate}\ m\ x))\}\\
\texttt{last} &= \lambda lst.\case\ lst \of \{(x:xs) \casepf\\
  & \quad\ \  \case\ xs \of \{ ([] \casepf x)\\
  & \quad\ \ \qquad\qquad~((y:ys) \casepf \texttt{last}\ xs)\}\}\\
 & \\
\end{array}\]
\end{minipage}
&
\begin{minipage}{.49\textwidth}
\vspace*{-0.35cm}\[\begin{array}{@{}l@{~}l@{}}
\reverse &= \lambda xs.\case\ xs \of \{ ([] \casepf [])\\
  & \quad\ \ ((y:ys) \casepf \reverse\ ys\ \texttt{++}\ [y])\}\\
\reverselin &= \lambda xs.\reverselinw\ []\ xs\\
\reverselinw &= \lambda xs,ys.\case\ ys \of \{ ([] \casepf xs)\\
  & \quad ((z:zs) \casepf \reverselinw\ (z:xs)\ zs)\}\\
\texttt{(++)} &= \lambda xs,ys.\case\ xs \of \{ ([] \casepf ys)\\
  & \quad\ \  ((z:zs) \casepf z:(zs\ \texttt{++}\ ys))\}\\
\texttt{concat} &= \lambda xs.(\texttt{foldr}\\
    & \quad\  (\lambda x,y.\texttt{foldr}\ (\lambda z,zs.(z:zs))\ y\ x)\\
    & \quad\ \qquad\qquad\quad\ \ []\ xs)\\
\texttt{concatMap} &= \lambda f,xs.(\texttt{foldr}\\
    & \quad\ \ \qquad\ \ \ (\lambda x,b.\texttt{foldr}\\
    & \quad\ \ \qquad\quad \ \ (\lambda z,zs.(z:zs))\ b\ (f\ x))\\
    & \quad\ \ \qquad\quad []\ xs)\\
\texttt{xor} &= \lambda x,y.\case\ x \of \{\\
    & \quad\ \ \qquad (\True \casepf \case\ y \of \{\\
    & \quad\ \ \qquad\qquad\quad( \True \casepf \False )\\
    & \quad\ \ \qquad\qquad\quad( \False \casepf \True )\})\\
    & \quad\ \ \qquad( \False \casepf y) \}\\
\end{array}\]
\end{minipage}
\\\hline
\end{tabular}
\caption{Several function definitions\label{fig:code}}
\end{figure*}
\begin{figure}[thb]
\[\arraycolsep=7pt
\begin{array}{l|rrrrrrrrrr}
k & 25 & 50 & 75 & 100 & 125 & 150 & 175 & 200 & 225 & 250\\
\hline
 & \multicolumn{10}{c}{\foldl{}\text{ using }\texttt{xor}}\\
\hline
\mln     & 302 & 602 & 902 & 1202 & 1502 & 1802 & 2102 & 2402 & 2702 & 3002\\
\mlnall  & 1085 & 2160 & 3235 & 4310 & 5385 & 6460 & 7535 & 8610 & 9685 & 10760\\
\spmax{} & 217 & 417 & 617 & 817 & 1017 & 1217 & 1417 & 1617 & 1817 & 2017\\
\hline
 & \multicolumn{10}{c}{\foldls{}\text{ using }\texttt{xor}}\\
\hline
\mln     & 327 & 652 & 977 & 1302 & 1627 & 1952 & 2277 & 2602 & 2927 & 3252\\
\mlnall  & 1235 & 2460 & 3685 & 4910 & 6135 & 7360 & 8585 & 9810 & 11035 & 12260\\
\spmax{} & 87 & 112 & 137 & 162 & 187 & 212 & 237 & 262 & 287 & 312\\
\hline
 & \multicolumn{10}{c}{\foldr{}\text{ using }\texttt{xor}}\\
\hline
\mln     & 279 & 554 & 829 & 1104 & 1379 & 1654 & 1929 & 2204 & 2479 & 2754\\
\mlnall  & 1016 & 2016 & 3016 & 4016 & 5016 & 6016 & 7016 & 8016 & 9016 & 10016\\
\spmax{} & 90 & 115 & 140 & 165 & 190 & 215 & 240 & 265 & 290 & 315
\end{array}\]
 \caption{Table of analysis results for different \texttt{fold}-variants}\label{fig:foldl-results}
 \end{figure} 
Following \cite{ifl16},  we use the \LRPi{} to find an example in which \foldl{} is worse than \foldr{} if the preconditions on arguments
 are not fulfilled. 
Choosing  \texttt{xor} for $f$ and \texttt{False} as $e$,
the requirement $f\ e\ s\preceq f\ s\ e$ holds,
but the requirement $(f\ (f\ s_1\ s_2)\ s_3) \preceq (f\ s_1\ (f\ s_2\ s_3))$ is not fulfilled for $s_1=\True$, $s_2=\False,s_3=\False$. 
A list starting with a single \True{} element followed by $k-1$ \False{}-elements generated using a take-function/list generator approach 
(using a Peano encoding to represent the numbers) is used as input list.

We configure \LRPi{} to collect garbage whenever possible. As we will see,  
\foldr{} indeed has a better runtime behavior than \foldl{} and the space consumption 
of \foldr{} and \foldls{} are almost equal. Moreover, we see that \texttt{foldl} has a much worse space behavior than \texttt{foldl'}. 
This difference is caused by the known 
stack problems of \foldl{} that can be solved in the case of \texttt{xor} by using \texttt{foldl'} instead.

We can identify the stack overflow problem (of fold-expressions) in the space diagram in Fig.~\ref{fig:fold-results-2} using $k=250$, directly calculated by  \LRPi{}. Let $s_i$ be the $i$-th expression during execution. 
Because of  lazy evaluation, the \foldl{}-expression is expanded step by step
 without calculating any intermediate results until \foldl{} itself is no longer required and is removed by the garbage collector. 
 This leaves a long chain of nested \texttt{(++)}-function calls that lead to the big rise of the curve, 
 because this causes a long chain of (lbeta)- and (case)-transformations. The small decrease before the rise of the curve is caused by the removal 
 of \foldl{} by the garbage collector, because the definition of \foldl{} is not needed anymore after the expansion is completed. 
 Note that (gc)-reductions are not counted by \mlnall{}, but counted in the following diagrams in Fig. \ref{fig:fold-results-2}.
\begin{figure}[htb]
\begin{center}\begin{tikzpicture}[>=stealth]
\begin{axis}[
  width=0.75\textwidth,
  height=5cm,
  axis x line=center, axis y line=center,
  xtick=\empty,
  ytick=\empty,
  extra x ticks={7,13,19,25},
  extra x tick labels={2369,4737,7106,9474},
  extra y ticks={2,4,6,8,10,12},
  extra y tick labels={288,576,864,1153,1441,1729},
  xlabel={$i$}, ylabel={$\size(s_i)$}, xlabel style={right}, ylabel style={above}, xmin=1, xmax=27, ymin=0, ymax=14, grid=major, samples=50,
  after end axis/.code={\path (axis description cs:0,0) node [anchor=north west,yshift=-0.075cm,xshift=-0.2cm] {1} node [anchor=south east,xshift=-0.075cm,yshift=-0.2cm] {0};}
]
  \addplot[mark=none] coordinates { (1.0,2.1) (1.1,2.1) (1.2,2.1) (1.3,2.1) (1.4,2.1) (1.5,2.1) (1.6,2.2) (1.7,2.2) (1.8,2.2) (1.9,2.2) (2.0,2.2) (2.1,2.2) (2.2,2.2) (2.3,2.2) (2.4,2.2) (2.5,2.2) (2.6,2.3) (2.7,2.3) (2.8,2.3) (2.9,2.3) (3.0,2.3) (3.1,2.3) (3.2,2.3) (3.3,2.3) (3.4,2.3) (3.5,2.3) (3.6,2.4) (3.7,2.4) (3.8,2.4) (3.9,2.4) (4.0,2.4) (4.1,2.4) (4.2,2.4) (4.3,2.4) (4.4,2.4) (4.5,2.4) (4.6,2.5) (4.7,2.5) (4.8,2.5) (4.9,2.5) (5.0,2.5) (5.1,2.5) (5.2,2.5) (5.3,2.5) (5.4,2.5) (5.5,2.5) (5.6,2.6) (5.7,2.5) (5.8,2.6) (5.9,2.6) (6.0,2.6) (6.1,2.6) (6.2,2.6) (6.3,2.6) (6.4,2.6) (6.5,2.6) (6.6,2.6) (6.7,2.7) (6.8,2.7) (6.9,2.7) (7.0,2.7) (7.1,2.7) (7.2,2.7) (7.3,2.7) (7.4,2.7) (7.5,2.7) (7.6,2.7) (7.7,2.7) (7.8,2.8) (7.9,2.8) (8.0,2.8) (8.1,2.8) (8.2,2.8) (8.3,2.8) (8.4,2.8) (8.5,2.8) (8.6,2.8) (8.7,2.9) (8.8,2.9) (8.9,2.9) (9.0,2.9) (9.1,2.9) (9.2,2.9) (9.3,2.9) (9.4,2.9) (9.5,2.9) (9.6,2.9) (9.7,3.0) (9.8,3.0) (9.9,3.0) (10.0,3.0) (10.1,3.0) (10.2,3.0) (10.3,3.0) (10.4,3.0) (10.5,3.0) (10.6,3.0) (10.7,3.1) (10.8,3.0) (10.9,3.1) (11.0,3.1) (11.1,3.1) (11.2,3.1) (11.3,3.1) (11.4,3.1) (11.5,3.1) (11.6,3.1) (11.7,3.1) (11.8,3.1) (11.9,3.2) (12.0,3.2) (12.1,3.2) (12.2,3.2) (12.3,3.2) (12.4,3.2) (12.5,3.2) (12.6,3.2) (12.7,3.2) (12.8,3.2) (12.9,3.3) (13.0,3.3) (13.1,3.3) (13.2,3.3) (13.3,3.3) (13.4,3.3) (13.5,3.3) (13.6,3.3) (13.7,3.3) (13.8,3.3) (13.9,3.4) (14.0,3.4) (14.1,3.4) (14.2,3.4) (14.3,3.4) (14.4,3.4) (14.5,3.4) (14.6,3.4) (14.7,3.4) (14.8,3.5) (14.9,3.5) (15.0,3.5) (15.1,3.5) (15.2,3.5) (15.3,3.5) (15.4,3.5) (15.5,3.5) (15.6,3.5) (15.7,3.5) (15.8,3.6) (15.9,3.5) (16.0,3.6) (16.1,3.6) (16.2,3.6) (16.3,3.6) (16.4,3.6) (16.5,3.6) (16.6,3.6) (16.7,3.6) (16.8,3.7) (16.9,3.6) (17.0,3.7) (17.1,3.7) (17.2,3.7) (17.3,3.7) (17.4,3.7) (17.5,3.7) (17.6,3.7) (17.7,3.7) (17.8,3.7) (17.9,3.7) (18.0,3.8) (18.1,3.8) (18.2,3.8) (18.3,3.8) (18.4,3.8) (18.5,3.8) (18.6,3.8) (18.7,3.8) (18.8,3.7) (18.9,3.9) (19.0,4.1) (19.1,4.3) (19.2,4.5) (19.3,4.7) (19.4,4.9) (19.5,5.1) (19.6,5.3) (19.7,5.5) (19.8,5.7) (19.9,5.9) (20.0,6.1) (20.1,6.3) (20.2,6.5) (20.3,6.7) (20.4,6.9) (20.5,7.2) (20.6,7.5) (20.7,7.7) (20.8,7.9) (20.9,8.1) (21.0,8.3) (21.1,8.5) (21.2,8.7) (21.3,8.9) (21.4,9.1) (21.5,9.3) (21.6,9.5) (21.7,9.7) (21.8,9.9) (21.9,10.1) (22.0,10.3) (22.1,10.5) (22.2,10.7) (22.3,10.9) (22.4,11.2) (22.5,11.4) (22.6,11.6) (22.7,11.8) (22.8,12.0) (22.9,12.2) (23.0,12.4) (23.1,12.6) (23.2,12.8) (23.3,13.0) (23.4,13.2) (23.5,13.4) (23.6,13.6) (23.7,13.8) (23.8,14.0) (23.9,13.7) (24.0,13.3) (24.1,12.8) (24.2,12.4) (24.3,11.9) (24.4,11.5) (24.5,11.1) (24.6,10.6) (24.7,10.2) (24.8,9.7) (24.9,9.3) (25.0,8.9) (25.1,8.4) (25.2,8.0) (25.3,7.6) (25.4,7.1) (25.5,6.6) (25.6,6.1) (25.7,5.7) (25.8,5.3) (25.9,4.8) (26.0,4.4) (26.1,4.0) (26.2,3.5) (26.3,3.1) (26.4,2.6) (26.5,2.2) (26.6,1.8) (26.7,1.3) (26.8,0.9) (26.9,0.4) (27.0,1.0e-2) };
\end{axis}
\end{tikzpicture}\end{center}
\caption{Size diagram for \texttt{foldl} using \texttt{xor} and input size $k=250$} \label{fig:fold-results-2}
\end{figure}
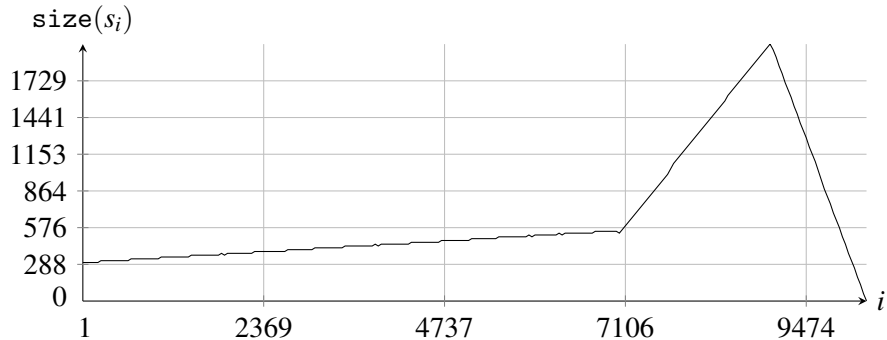

\begin{figure}[htb]
\[\arraycolsep=8pt
\begin{array}{l|rrrrrrrr}
k & 50 & 100 & 150 & 200 & 250 & 300 & 350 & 400\\
\hline
 & \multicolumn{8}{c}{\texttt{last (reverse (replicate }k\texttt{ True))}}\\
\hline
\mln     & 4230 & 15955 & 35180 & 61905 & 96130 & 137855 & 187080 & 243805\\
\mlnall  & 15799 & 59074 & 129849 & 228124 & 353899 & 507174 & 687949 & 896224\\
\spmax{} & 462 & 862 & 1262 & 1662 & 2062 & 2462 & 2862 & 3262\\
\hline
 & \multicolumn{8}{c}{\texttt{last (reverse' (replicate }k\texttt{ True))}}\\
\hline
\mln     & 457 & 907 & 1357 & 1807 & 2257 & 2707 & 3157 & 3607\\
\mlnall  & 1782 & 3532 & 5282 & 7032 & 8782 & 10532 & 12282 & 14032\\
\spmax{} & 100 & 150 & 200 & 250 & 300 & 350 & 400 & 450\\
\end{array}\]
\caption{Comparing two reverse variants} \label{fig:reverse-experiments}
\end{figure}

We now want to compare $\reverse{}$ with $\reverselin{}$    
in Fig \ref{fig:reverse-experiments}. We use \texttt{last} to force the evaluation and moreover we create a list
 containing $k$ times the element \True{} using \texttt{replicate }$k$\texttt{ True}.
This supports the following conjectures on  complexities: $\reverse{}$ requires quadratic runtime, caused by the left-associativity of \texttt{(++)} 
while $\reverselin{}$ 
requires linear runtime. Because \texttt{(++)} only goes through each intermediate list, \reverse{} 
appears to not need asymptotically more space than $\reverselin{}$. Both $\reverse{}$ and $\reverselin{}$ 
appear to have a linear space complexity, perhaps $\reverselin{}$ has smaller constants in the asymptotic complexity formula.

Now we want to have a short look on fusion. The composition of functions can lead to well readable programs, because recursions are hidden and the main steps of the calculation
 are clearly visible. But this leads to intermediate structures and to an increase of the reduction length and especially space consumption, 
 if we use a realistic (non-eager) garbage collector.
 The Glasgow Haskell Compiler (GHC) uses the so called short cut fusion as introduced in \cite{gill}. This approach eliminates such intermediate tree and list structures to gain a better runtime and to reduce the needed space.

As shown in \cite{johann}, short cut fusion might be unsafe if \texttt{seq} is used,  but in the majority of cases  this approach works and is used by the GHC. 
Moreover \cite{svenningsson} shows that this approach might increase sharing and therefore a part of the memory is longer used. Thus it may increase the space consumption. 

We now want to compare \texttt{(comp concat map) tail} with \texttt{concatMap tail}. As input we use a list containing $k$ inner lists of the form \texttt{[True,True]}, 
again generated by a list-generator/take-function approach. 
The   differences in the table are the unfused version minus the fused version. The results are in Fig.~\ref{fig:fusion-results}.
\begin{figure}
\[\arraycolsep=7pt\begin{array}{l|rrrrrrrrrr}
k & 100 & 200 & 300 & 400 & 500 & 600 & 700 & 800 & 900 & 1000\\
\hline
 & \multicolumn{10}{c}{\text{Difference of reduction lengths between fused and unfused}}\\
\hline
\Delta\ \mln{} & 206 & 406 & 606 & 806 & 1006 & 1206 & 1406 & 1606 & 1806 & 2006\\
\Delta\ \mlnall{} & 623 & 1223 & 1823 & 2423 & 3023 & 3623 & 4223 & 4823 & 5423 & 6023\\
\hline
 & \multicolumn{10}{c}{\text{Difference of }\spmax{}\text{ between fused and unfused}}\\
\hline
\Delta\ \text{Eager}               & 14 & 14 & 14 & 14 & 14 & 14 & 14 & 14 & 14 & 14\\
\Delta\ \text{Every }1000\text{th} & 47 & 47 & 47 & 47 & 47 & 47 & 47 & 47 & 47 & 47\\
\Delta\ \text{Every }2000\text{th} & 60 & 60 & 60 & 60 & 60 & 60 & 60 & 60 & 60 & 60\\
\Delta\ \text{Never}               & 132 & 232 & 332 & 432 & 532 & 632 & 732 & 832 & 932 & 1032
\end{array}\]
\caption{Differences in time and space between fused and unfused \texttt{concatMap}} \label{fig:fusion-results}
\end{figure}
As expected the reduction length and space consumption behaves linearly in all cases. We also see that the frequency of the garbage collector directly affects the space consumption, 
if we compare each garbage collection mode of the fused with the unfused version. 
The rarer the garbage collector runs the higher is the difference in space consumption: If we turn off the garbage collector and use the fused version instead of the unfused version, then the decrease of space consumption is linear in the length of the list.

With regard to $nogc$ the advantage concerning space consumption of the fused versions over the unfused versions of the above examples is only constant, but the advantage is even linear if we turn off garbage collection. 
Thus the above examples for fusion are space improvements in a weak sense. 
Practically, the weak space improvements above are very useful because they are also time improvements.   

\begin{figure}
\[\arraycolsep=8pt
\begin{array}{l|rrr|rrrrr}
k & 12 & 13 & 14 & 200 & 400 & 600 & 800 & 1000\\
\hline
 & \multicolumn{8}{c}{\text{Shared append}}\\
\hline
\mln     & 297 & 321 & 345 & 4809 & 9609 & 14409 & 19209 & 24009\\
\mlnall  & 1152 & 1245 & 1338 & 18636 & 37236 & 55836 & 74436 & 93036\\
\spmax{} & 77 & 79 & 81 & 453 & 853 & 1253 & 1653 & 2053\\
\hline
 & \multicolumn{8}{c}{\text{Unshared append}}\\
\hline
\mln     & 453 & 489 & 525 & 7221 & 14421 & 21621 & 28821 & 36021\\
\mlnall  & 1730 & 1867 & 2004 & 27486 & 54886 & 82286 & 109686 & 137086\\
\spmax{} & 78 & 79 & 80 & 266 & 466 & 666 & 866 & 1066\\
\end{array}\]
\caption{Shared versus unshared append} \label{fig:append-results}
\end{figure}

The final example is a case where the decrease of space consumption behaves inverse to time consumption.
The example experiments in Fig. \ref{fig:append-results} reports on 
comparing   $(list~\texttt{++}~list)$ ~$\texttt{++}$~ $(list~\texttt{++} ~list)$    
with  
${\tt let}~xs =  list$ ~
 ${\tt in}$ ~$(xs~\texttt{++}~xs)~\texttt{++}$~ $(xs~ \texttt{++}~ xs)$ (written here in Haskell notation), 
 driving evaluation using the ${\tt last}$ function, 
 and where \texttt{++} is the the ${\tt append}$ function.
The first expression has four separate occurrences of a (long) ${list}$, whereas the second
 expression shares the $list$s, where ${list}$  varies in length in the experiments. 
 The results are consistent with the claim that common subexpression elimination (cse) is a time improvement \cite{schmidt-schauss-sabel-PPDP:2015}, 
 and show that (cse) and an increase of sharing in general
 may increase the (maximal) space usage. In neither direction the example is a space improvement, which shows that (cse) is not a space improvement.

\section{Conclusion and Future Work} \label{sec:conc}
We demonstrated that the  interpreter \LRPi{} is a useful tool for exploring improvements.
The conceptual work on it also had an influence on constructing appropriate models of resource consumption.
Among the influences are: the calculus must incorporate  (gc), and  the Sestoft machine turned out to have a non-optimal space behavior,
which had to be improved.
 We expect that in the future there will be more influences and feedback in both directions between measuring tool with its
 experiments and the theory.

Future research into the relations between calculus, machine translations and  abstract machine is justified.
Further work is to extend \LRPi{}  also taking contexts (according to Def. \ref{def:spmax-improvement}) into account, 
or automating the inspection of series of arguments, 
in order to improve the  affirmative power for space improvements. Moreover, a more practical integer representation would be helpful, 
since Peano encodings affect and pollute the space measurement.
Also refining the garbage collection (for example locally generated garbage) is an issue.

\subsection*{Acknowledgments}  We thank David Sabel for discussions and hints which were very helpful in improving the paper.
We also thank the reviewers of WPTE for the numerous helpful remarks. 
 \bibliographystyle{eptcs}

\end{document}